\documentclass[11pt]{article}
\usepackage{amsmath,amssymb,amsthm,amsfonts}
\usepackage{latexsym}
\usepackage{epsfig}
\usepackage[right=0.8in, top=1in, bottom=1in, left=0.8in]{geometry}
\usepackage{setspace}
\usepackage{bbold}
\usepackage{mathtools}
\usepackage{dirtytalk}
\usepackage{enumerate}
\usepackage{color,soul}
\usepackage{stackengine}
\usepackage{algpseudocode}
\usepackage{multicol}
\usepackage{extarrows}
\usepackage[linesnumbered,ruled,vlined]{algorithm2e}
\usepackage{verbatim}
\usepackage[backref,colorlinks,citecolor=blue,bookmarks=true]{hyperref}
\usepackage[nameinlink]{cleveref}
\usepackage[normalem]{ulem}

\usepackage{scalerel,stackengine}
\stackMath
\newcommand\reallywidehat[1]{%
\savestack{\tmpbox}{\stretchto{%
  \scaleto{%
    \scalerel*[\widthof{\ensuremath{#1}}]{\kern-.6pt\bigwedge\kern-.6pt}%
    {\rule[-\textheight/2]{1ex}{\textheight}}
  }{\textheight}%
}{0.5ex}}%
\stackon[1pt]{#1}{\tmpbox}%
}

\newcommand{\N}{{\mathbb{N}}}

\newcommand{\R}{{\mathbb{R}}}

\newcommand{\A}{{\mathcal{A}}}

\newtheorem{theorem}{Theorem}
\newtheorem*{theorem*}{Theorem}

\newtheorem{lemma}[theorem]{Lemma}

\newtheorem{definition}[theorem]{Definition}
\newtheorem{claim}[theorem]{Claim}
\newtheorem{fact}[theorem]{Fact}

\makeatletter
\newtheorem*{rep@theorem}{\rep@title}
\newcommand{\newreptheorem}[2]{%
\newenvironment{rep#1}[1]{%
 \def\rep@title{#2 \ref{##1}}%
 \begin{rep@theorem}}%
 {\end{rep@theorem}}}
\makeatother
\newreptheorem{theorem}{Theorem}

\newcommand{\norm}[1]{\left\lVert#1\right\rVert}

\usepackage[pdftex,dvipsnames]{xcolor}  
\usepackage[utf8]{inputenc}
\usepackage{amsmath}
\crefname{ineq}{inequality}{inequalities}
\creflabelformat{ineq}{#2{\upshape(#1)}#3}
\usepackage{cryptocode}   
\graphicspath{ {img/} }
\usepackage[colorinlistoftodos,prependcaption,textsize=tiny]{todonotes}
\newcommandx{\unsure}[2][1=]{\todo[linecolor=red,backgroundcolor=red!25,bordercolor=red,#1]{#2}}
\newcommandx{\change}[2][1=]{\todo[linecolor=blue,backgroundcolor=blue!25,bordercolor=blue,#1]{#2}}
\newcommandx{\info}[2][1=]{\todo[linecolor=OliveGreen,backgroundcolor=OliveGreen!25,bordercolor=OliveGreen,#1]{#2}}
\newcommandx{\improvement}[2][1=]{\todo[linecolor=Plum,backgroundcolor=Plum!25,bordercolor=Plum,#1]{#2}}
\newcommandx{\thiswillnotshow}[2][1=]{\todo[disable,#1]{#2}}

\newcommand{\remove}[1]{}

\renewcommand{\bf}{\normalfont \bfseries}

\renewcommand{\paragraph}[1]{\vspace{0.7ex}\noindent{\bf #1}}

\newcommand{\LCS}[0]{\textsc{LCS}\xspace}
\newcommand{\LIS}[0]{\textsc{LIS}\xspace}

\newcommand{\OPT}[0]{\textsf{OPT}\xspace}

\newcommand{\1}[0]{\mathbb{1}}

\newcommand{\abl}[0]{{\cA_{\text{BL}}}}

\def\1{{\mathbf{1}}}

\newcommand{\cA}{\mathcal{A}}

\def\colorful{1}

\ifnum\colorful=1

\fi
\ifnum\colorful=0

\fi

\newcommand{\ignore}[1]{{}}



\newcommand{\reslis}[0]{\textsf{Block-LIS}\xspace}

\newcommand{\blocklis}[0]{{\textsf{Block-LIS}}}

\title{
Approximating the Longest Common Subsequence problem within a sub-polynomial factor in linear time
}

\author{
Negev Shekel Nosatzki\\Columbia University\\\texttt{ns3049@columbia.edu}}

\date{}

\begin{document} 
\maketitle
\thispagestyle{empty}

\begin{abstract}
	The Longest Common Subsequence (\LCS) of two strings is a fundamental string similarity measure with a classical dynamic programming solution taking quadratic time. Despite significant efforts, little progress was made in improving the runtime. Even in the realm of approximation, not much was known for linear time algorithms beyond the trivial $\sqrt{n}$-approximation. Recent breakthrough result provided a $n^{0.497}$-factor approximation algorithm \cite{doi:10.1137/1.9781611975482.72}, which was more recently improved to a $n^{0.4}$-factor one\cite{DBLP:journals/corr/abs-2106-08195}.
The latter paper also showed a $n^{2-2.5\alpha}$ time algorithm which outputs a $n^{\alpha}$ approximation to the \LCS, but so far no sub-polynomial approximation is known in truly subquadratic time.

		In this work, we show an algorithm which runs in $O(n)$ time, and outputs a $n^{o(1)}$-factor approximation to $\LCS(x,y)$, with high probability, for any pair of length $n$ input strings.
		
		Our entire algorithm is merely an efficient black-box reduction to the $\reslis$ problem, introduced very recently in \cite{andoni2021estimating}, and solving the $\reslis$ problem directly.
	 	\end{abstract} 
	 	
		\newpage
\setcounter{page}{1}

\section{Introduction}

The Longest Common Subsequence problem (\LCS) is a classical string alignment problem. Formally, given two strings $x,y$ of length $n$ over alphabet  $\Sigma$, the longest common subsequence between $x$ and $y$ is the largest string which appears as a subsequence in both $x$ and $y$. 
The (length of the) \LCS has a classical dynamic programming solution in quadratic time \cite{WF74}. 

\LCS is also a fundamental string similarity measure of key importance both in theory and practice, and therefore has been extensively studied during the course of the past 50 years \cite{ullman1976bounds,10.1145/322033.322044,10.1145/359581.359603,apostolico1987longest,10.1145/146637.146650,SunWoodruff_LCS_LIS,rubinstein2019approximation,rubinstein2020reducing}, with the main goal of improving the run-time from
polynomial towards close(r) to linear. 

Despite significant attention, little progress was made, and the runtime has improved by logarithmic factors only \cite{MP80}. More recently, fine-grained hardness results helped explaining the lack of progress, showing that any truly
subquadratic solution for LCS would refute the Strong Exponential Time Hypothesis (SETH)\cite{7354388,bringmann2015quadratic}. Moreover, for the \LCS problem in particular, it was later shown that even a high polylogarithmic speedup would be surprising  \cite{abboud2016simulating,abboud2018tighter}.

Even before the above hardness results, researchers started considering faster algorithms that approximate the \LCS, but there as well the progress has been slow. The trivial $\sqrt{n}$-factor approximation was only recently improved to $\approx n^{0.497}$ \cite{doi:10.1137/1.9781611975482.72} and even more recently to $n^{0.4}$ \cite{DBLP:journals/corr/abs-2106-08195}.
 In parallel, a parametrized version was developed which gives a subquadratic solution with improved approximation when the longest common subsequence is promised to be above a certain threshold \cite{rubinstein2019approximation}.
 
In contrast, the progress for the corresponding Edit Distance problem has been more ``rosy'', with a series of advances \cite{BJKK04,BES06,AKO-edit,chakraborty2018approximating}, recently culminating in a near-linear time, constant factor approximation algorithm \cite{asn20edit}.

\subsection{Main result}

Our main result is a sub-polynomial factor approximation algorithm for \LCS, which runs in linear time.

\begin{theorem}\label{thm::lcs_main}
	$\LCS_n$ can be approximated up to $n^{o(1)}$-factor, with high probability, and in $O(n)$ time.
\end{theorem}

While we do not track the exact approximation factor in our analysis, we note it is $\exp\left(\tfrac{\log(n)}{\log^{c}\log(n)}\right)$ for some absolute constant $c$. Interesting enough, we could not find a way to significantly improve the bound above, obtainable in linear time, even if one can afford $n^{1.99}$ runtime. This relates to the issue of balancing of dense and sparse cases described in \cite{andoni2021estimating}, which is the core component of our algorithm.

To improve the bound, one would need to find a more efficient way to solve the \reslis algorithm from \cite{andoni2021estimating}. 

One of the implications of this paper is that an affirmative resolution to the open question posed in \cite{andoni2021estimating} would yield surprising results, and in particular, a near-linear time, $(1+\epsilon)$ approximation to $\LCS$, which also provides near-linear $\epsilon n$ additive approximation to the edit distance.

\section{Preliminaries}

\subsection{Definitions}

\paragraph{Sequences.} Fix alphabet $\Sigma$. We denote $x \in \Sigma^n$ as a length $n$ sequence.

For a sequence $x$, we denote by $\psi(x): \Sigma \rightarrow \N$, the count vector of letters in $x$, i.e., $\psi(x)_c \triangleq |x^{-1}(c)|$ for any $c \in \Sigma$.

We also define integer block sequence $z \in \N^{n \times k}$, and denote $|z|$ as the count of all integers in all the blocks (with repetitions).

\paragraph{Monotone sets.} 
We
 define monotone sets as follows:
\begin{definition}
	We say that a set 
	$P \subseteq \N \times \N$ is \emph{monotone} iff
	for all $((i,k),(j,l)) \in P \times P$, we have (i) $i = j \Leftrightarrow l=k$; and (ii) $i < j \Leftrightarrow k < l$.
\end{definition}

One can observe that 
$$\LCS(x,y) = 
	\max_{\stackrel{P \subseteq [n] \times [n]}{P \text{ is monotone}}} \sum_{(i,j) \in P} \1[x_i = y_j]
$$

\paragraph{The \reslis problem.}
The $\reslis$ problem was defined in \cite{andoni2021estimating}\footnote{The original definition of $\reslis$ included another extension which is not relevant for our reduction, and hence was omitted for exposition.}. In $\reslis$, the main input consists of $n$ blocks of (at most) $k$ elements each, and each block can contribute 
at most one of its elements to a subsequence. Formally, for an integer block sequence 
$z \in \N^{n \times k}$, $\reslis(z)$ is the length of a maximal sub-sequence $\OPT \triangleq \{(w_1,z_w)\}_{w \in [n] \times [k]}$ 
such that $\OPT$ is monotone.

\paragraph{Other notations.} Fix vectors $a,b \in \R^n$, then $\min(a,b)$ (respectively $\max(a,b)$) is defined, as the vector of minimums (respectively, maximums) for each $(a_i,b_i)$ pair.
Also, the notation $\oplus$ means the direct sum.

\subsection{Main technique - \reslis reduction}

Our construction boils down to a formulation of the \LCS problem as a (larger) \LIS problem, and solve the \LIS problem directly. Our reduction at its core is reminiscent to ideas introduced in \cite{10.1145/359581.359603,apostolico1987longest}, with two important differences:

\begin{enumerate}
	\item We provide a simplified approach by direct ``black-box'' reduction to a $\LIS$ problem which is then solved directly. In particular, we reduce to the \reslis problem. As a result, the entire construction and proof takes roughly 1.5 pages.
	\item  We show that the number of ``matching pairs'' is always bounded by $2n \cdot \LCS(x,y)$ using norm inequalities, implying the linear runtime from Lemma 4.3 in \cite{andoni2021estimating}.
\end{enumerate}

In particular, we restate the following Definition and Theorem:

\begin{definition}\cite{andoni2021estimating}
	For $\alpha \geq 1$ 
	and $\beta>0$, an $(\alpha,\beta)$-approximation $\hat{q}$ of a quantity $q$ is an $\alpha$-multiplicative and $\beta$ additive estimation of $q$, i.e., $\hat{q} \in [q/\alpha - \beta, q]$.
\end{definition}

\begin{theorem}\cite{andoni2021estimating}\label{thm::reslis_main}
	Fix $n,k \in \N$ and $z \in \N^{n \times k}$. 
	For any $\lambda = o(1)$, there exists a randomized non-adaptive algorithm $\abl$ that solves $\reslis$ up to $(\alpha,\lambda n)$-approximation, where $\alpha=\left(\tfrac{|z|}{\lambda n}\right)^{o(1)}$ using $\left(\tfrac{|z|}{\lambda n}\right)\cdot n^{o(1)}$ time.
\end{theorem}

\paragraph{\reslis Formulation for LCS.}
To prove \Cref{thm::lcs_main}, we formulate the $\LCS$ problem as a single string $\blocklis$.
Given $x,y \in \Sigma^n$, we define an integer block sequence $z$ such that each $z_i$ is a block of all $j \in [n]$ indices where $y_j = x_i$. Formally:
\begin{equation}\label{eq::z_def}
z \triangleq \oplus_{i \in [n]} \left\{y^{-1}\left(x_i\right) \right\}.	
\end{equation}

We claim the following equivalence:
\begin{claim}
$\reslis(z) = \LCS(x,y)$.	
\end{claim}
\begin{proof}
	First, observe that any common subsequence of $x$ and $y$ has a monotone set $P \subseteq [n] \times [n]$ where $x_i = y_j$ for all $(i,j) \in P$. Now, since $j \in y^{-1}(x_i)$, then the second coordinates of $P$ are also an increasing subsequence of $z$.
	
	Similarly, consider the coordinates of an increasing subsequence $W \subseteq [n] \times [k]$. Then $\{(w_1, z_w)\}_{w \in W}$ is a monotone set where $x_{w_1} = y_{z_w}$, and hence $x$ and $y$ has a common subsequence of the same length.
\end{proof}
 Our entire algorithm is merely an efficient implementation of the above reduction, and solving it using \Cref{thm::reslis_main}. In particular, we show the following:

\begin{lemma}\label{lm::lcs_lis_reduction}
	Suppose there exists an algorithm $\abl$ which $\left(f\left(\tfrac{|z|}{\lambda n}\right),\lambda n\right)$ approximates \reslis in time $t
	\left(\tfrac{|z|}{\lambda n},n\right)$. Then, there exists an algorithm $\A_{\LCS}$ approximating $\LCS_n$ up to a factor $(1+o(1)) \cdot f(n \cdot \log n)$ in $t\left(O(n \cdot \log(n)),n\right) + O(n)$ time.
\end{lemma}

\begin{proof}[Proof of \Cref{thm::lcs_main} using \Cref{thm::reslis_main,lm::lcs_lis_reduction}]
	The proof 
	follow almost immediately from \Cref{thm::reslis_main,lm::lcs_lis_reduction}. However, a direct invocation of the statements only yields $n^{1+o(1)}$ time algorithm. To obtain linear time, we first subsample $U \subseteq [n]$ i.i.d. at rate $n^{-o(1)}$, and then apply the statements on $x_U$ and $y_U$, noting $\LCS(x_U,y_U) \in [n^{-o(1)},1] \cdot \LCS(x,y)$. The resulting runtime is $|U|^{1 + o(1)} = O(n)$ as needed.
\end{proof}

\subsection{Useful lower bound for the \LCS using Hölder's inequality}

\paragraph{Hölder's inequality.} A classical norm inequality which says that $\langle x,y \rangle \leq \norm{x}_{1/p} \cdot \norm{y}_{1/q}$ for any $p + q = 1$. In particular, $\langle x,y \rangle \leq \norm{x}_{\infty} \cdot \norm{y}_{1}$. Such inequality can be extended to show the following bound:

\begin{claim}\label{cl::min_holder}
	Fix $x,y \in \R^n$.
	Then, $\langle x,y \rangle \leq \norm{\min\{x,y\}}_{\infty} \cdot \left(\norm{x}_{1} + \norm{y}_{1}\right)$.
\end{claim}
\begin{proof}
	Let $a = \min\{x,y\}$ and $b = \max\{x,y\}$. Then, by Hölder's inequality, 
	$$\langle x,y \rangle = \langle a,b \rangle \leq \norm{a}_\infty \cdot \norm{b}_1 = \norm{\min\{x,y\}}_{\infty} \cdot \norm{\max\{x,y\}}_{1}\leq \norm{\min\{x,y\}}_{\infty} \cdot \left(\norm{x}_{1} + \norm{y}_{1}\right).$$ 
\end{proof}

\paragraph{Simple lower bound for the $\LCS$.}
We next state a useful fact about the minimal \LCS of two strings:

\begin{fact}\label{ft::min_lis}
 $\LCS(x,y) \geq \norm{\min\{\psi(x),\psi(y)\}}_\infty$.	
\end{fact}
\begin{proof}
	That fact is immediate, since if each of $x$ and $y$ contain $k$ copies of the same letter $c$, then $c^k$ is a common subsequence to both $x$ and $y$.
\end{proof}

We combine the above to show the following lower bound on the $\LCS$:
\begin{lemma}\label{lm::lcs_lower_bound}
 $\LCS(x,y) \geq \tfrac{|z|}{2n}$.	
\end{lemma}
\begin{proof}
	Observe that by construction $|z| = \langle\psi(x),\psi(y)\rangle$. Now, invoking \Cref{cl::min_holder,ft::min_lis}, we obtain:
	$$
	|z| \leq \norm{\min\{\psi(x),\psi(y)\}}_{\infty} \cdot \left(\norm{\psi(x)}_{1} + \norm{\psi(y)}_{1}\right) \leq 2n \cdot \LCS(x,y),
	$$
	which implies the Lemma.
\end{proof}

One direct implication of the above is a simple deterministic and exact algorithm for $\LCS$ in time $O(n \ell \log(\ell))$, where $\ell = \LCS(x,y)$. While we do not derive it in this paper, one can observe this can be done by simply reducing the $\LCS$ problem to $\reslis$, and solving it directly as a standard $\LIS$ problem, which takes $O(|z|\log(\LIS(z)))$ time.
\section{Reduction Algorithm}

In this section we prove \Cref{lm::lcs_lis_reduction}. Namely, we are given an algorithm $\abl$ for solving $\blocklis(z)$ with additive error parameter $\lambda$, and we use it to estimate $\LCS(x,y)$.

The algorithm is the straight-forward one: 

\vspace{2mm}
\begin{figure}[h!]
\fbox{
\begin{minipage}[t]{0.95\textwidth}
\textsc{EstimateLCS}$(x,y)$:
\begin{enumerate}
	\item Store the function $y^{-1}$ over $\Sigma$ (say, in the sparse matrix representation).
	\item 
	Set $d \gets 
	|z|/2n$, where $z$ is defined as in \Cref{eq::z_def}.
	\item Run
	$\abl$ with input $z$ and $\lambda \gets \tfrac{d}{n \log n}$. 
	\item Output $\widehat{\ell} \gets$ the maximum between $d$ and the output of $\abl$. 
\end{enumerate} 
\end{minipage}
}
\caption{Description of the algorithm for estimating the \LCS between $x$ and $y$}
\end{figure}
\subsection{Analysis}

We show the above algorithm provides the guarantees of \Cref{lm::lcs_lis_reduction}.

\begin{proof}[Proof of \Cref{lm::lcs_lis_reduction}]
	
First, let us argue $\widehat{\ell} \in \left[\tfrac{1-o(1)}{f(n \log n)},1\right] \cdot \LCS(x,y)$. 

\paragraph{Upper Bound.}
The upper-bound is immediate since the output of $\abl$ is promised to be at most $\reslis(z) = \LCS(x,y)$, and since $d \leq \LCS(x,y)$ from \Cref{lm::lcs_lower_bound}. 

\paragraph{Lower Bound.}
Let $\widehat{\ell_\A}$ be the output of $\abl$. Then, we are guaranteed 
$$\widehat{\ell_\A} \geq \tfrac{ \blocklis(z)}{f(|z|/\lambda)} - \lambda n = \tfrac{\LCS(x,y)}{f(n \log n)} - \tfrac{2d}{\log n}.$$
Therefore, $\widehat{\ell} = \max\{\widehat{\ell_\A},d\} \geq (1-o(1)) \cdot \tfrac{\LCS(x,y)}{f(n \log n)}$ as needed.

Last, we show the runtime is at most $t\left(O(n \cdot \log(n)),n\right) + O(n)$.

\paragraph{Runtime Analysis.} The runtime for the first 2 steps is clearly $O(n)$. It remains to show the runtime on the $\reslis$ algorithm, which is promised to be at most 
$$t\left(\tfrac{|z|}{\lambda n},n\right) = t\left(\tfrac{2d}{\lambda},n\right) = t\left(O(n \cdot \log(n)),n\right).$$

This concludes the proof.

\end{proof}

\section{Acknowledgements}
The author would like to thank Alexandr Andoni for useful discussion and comments on earlier drafts of this note.

\begin{flushleft}
\bibliography{andoni}{}
\bibliographystyle{alpha}
\end{flushleft}

\end{document}